\documentclass[10pt,twocolumn,twoside]{IEEEtran}
\usepackage{cite}
\usepackage{amsmath,amssymb,amsfonts}
\usepackage{textcomp}
\usepackage{verbatim}
\usepackage{titletoc}
\usepackage{caption}
\usepackage{subfigure}
\usepackage{algorithm}  
\usepackage{enumerate}
\usepackage{enumitem} 
\usepackage{algpseudocode}
\usepackage[hidelinks]{hyperref} 
\usepackage{wrapfig}
\usepackage{graphicx}
\usepackage{subfigure}
\usepackage{bm} 
\usepackage{amsthm}
\usepackage{array}
\usepackage{xcolor}
\usepackage{longtable}
\usepackage{multirow}
\usepackage{multicol}
\usepackage{fancyhdr}
\usepackage{setspace}
\usepackage{booktabs}
\usepackage{mathrsfs}
\usepackage[numbers,sort&compress]{natbib}

\allowdisplaybreaks[4]
\newtheorem{thm}{Theorem}[]
\newtheorem{lemma}{Lemma}[]

\newtheorem{remark}{Remark}[]

\newtheorem{definition}{Definition}[]

\def\BibTeX{{\rm B\kern-.05em{\sc i\kern-.025em b}\kern-.08em
    T\kern-.1667em\lower.7ex\hbox{E}\kern-.125emX}}

\begin{document}
\title{Bandwidth-Constrained Sensor Scheduling: A Trade-off between Fairness and Efficiency}
\author{Yuxing Zhong, Yuchi Wu, Daniel E.~Quevedo, \IEEEmembership{Fellow,~IEEE}, and Ling Shi, \IEEEmembership{Fellow,~IEEE}
\thanks{Yuxing Zhong and Ling Shi are with the Department of Electronic and Computer Engineering, Hong Kong University of Science and Technology, Hong Kong, China (e-mail:~\href{yuxing.zhong@connect.ust.hk}{yuxing.zhong@connect.ust.hk};~\href{eesling@ust.hk}{eesling@ust.hk}).}
\thanks{Yuchi Wu is with the School of Mechatronic Engineering and Automation, Shanghai University, Shanghai, China (e-mail:~\href{wuyuchi1992@gmail.com}{wuyuchi1992@gmail.com}).}
\thanks{Daniel E.~Quevedo is with the School of Electrical and Computer Engineering, the University of Sydney, Sydney, Australia (e-mail:~\href{daniel.quevedo@sydney.edu.au}{daniel.quevedo@sydney.edu.au}).}
}
\maketitle

\begin{abstract}
We address fair sensor scheduling over bandwidth-constrained communication channels. While existing literature on fair scheduling overlooks overall system efficiency, we introduce a novel $q$-fairness framework to balance efficiency and fairness by adjusting the parameter $q$. Specifically, for two communication scenarios, we: (i) derive the optimal schedule under limited communication rates, and (ii) propose two suboptimal algorithms under limited simultaneous sensor transmissions and analyze their performance gaps relative to the optimal strategy. Simulations demonstrate that our algorithms effectively balance efficiency and fairness in both cases.
\end{abstract}

\section{Introduction}
In wireless sensor networks, the sensors measure their environment and transmit relevant information to a remote state estimator for state estimation via wireless communications. These wireless sensor networks are widely used in applications like environmental monitoring and target tracking~\cite{yick2008wireless}. Although they offer benefits such as low cost and flexibility~\cite{mahmoud2018fundamental}, sensors are usually battery-powered and have limited communication bandwidth. Therefore, it is crucial to carefully schedule data transmissions to conserve resources.

There is an extensive body of literature on sensor scheduling problems. {It has been shown that the optimal scheduling strategy is periodic and exhibits a threshold-type property~\cite{han2017optimal,chakravorty2017fundamental,leong2016sensor}. Moreover, convex relaxation is widely used in determining scheduling strategies~\cite{maity2022sensor,vafaee2024learning}. However, these offline schedulers, while easy to implement, are less flexible and efficient than online approaches~\cite{astrom2002comparison}. In event-triggered scheduling, sensor transmissions are triggered based on real-time values such as measurements~\cite{han2015stochastic} and estimation error covariance~\cite{trimpe2014event}.  Research on sensor scheduling has also extended to scenarios involving lossy channels~\cite{zhong2023event,zhong2024event}, distributed systems~\cite{duan2022sensor}, and malicious attacks~\cite{liu2022rollout}.} However, the cost function in the above literature is the summation of the estimation errors across different systems. While minimizing this total cost yields efficient overall system performance, it often leads to unfair resource allocations. Specifically, some systems may enjoy low estimation errors while others suffer high errors.

When equal performance among individuals is desired, or when individuals prioritize their own benefits over overall system performance (i.e., lower total cost), it becomes crucial to equalize estimation errors across systems rather than simply minimizing their sum. This requires a fair scheduling algorithm. A key limitation of the aforementioned methods is that they do not explicitly address fairness, as their objective is to minimize the total cost without considering how the cost is distributed among systems. Modifying the summation into a weighted sum does not address the fairness issue due to the difficulty in selecting appropriate weights to equalize the cost. To the best of our knowledge, the first and only work on fair sensor scheduling was presented in~\cite{wu2020max}, which provides the ``fairest" (max-min fair) schedules~\cite{5461911}. However, this approach prioritizes fairness at the expense of overall system performance, which can be problematic when estimation accuracy is also critical. Therefore, a method that can balance efficiency (overall performance) and fairness is highly needed.

The present article builds on a significant body of work incorporating fairness considerations into algorithm design, where trade-offs between fairness and efficiency are achieved~\cite{li2019fair,khan2016fairness, mo2002fair}. Inspired by these studies, this article examines the bandwidth-constrained sensor scheduling problem for dynamical systems and aims to establish flexible trade-offs between efficiency and fairness. The main contributions are summarized as follows:
(i) We consider a $q$-fairness framework for remote state estimation in the sensor scheduling problem. Different from~\cite{wu2020max} which achieves maximum fairness, we achieve flexible trade-offs between efficiency and fairness by adjusting $q$.
(ii) We study two different types of bandwidth constraints, i.e., communication rate constraints and sensor activation constraints. For the rate constraints, we provide the optimal scheduling strategy. For activation constraints, we develop two algorithms obtained by the Markov decision process (MDP) and greedy heuristics and provide their performance gaps compared to the optimal strategy.

\emph{Notations:} The notation $\mathbb{R}^n$ and $\mathbb{R}^{n\times n}$ denote the set of $n$-dimensional vectors and $n$-by-$m$-dimensional matrices. The sets $\mathbb{Z}$, $\mathbb{N}$, and $\mathbb{N}_+$ represent the set of integers, natural numbers, and positive natural numbers, respectively. For a matrix $X$, $X^T$, ${\rm Tr}(X)$ and $\rho(X)$ are its transpose, trace and spectral radius, respectively. Positive semidefinite (definite) matrices are denoted by $X\succeq0$ ($X\succ0$). The notations $\bm 1$ and $I$ denote the all-ones vector and the identity matrix with compatible dimensions, respectively. The expectation of a random variable is denoted by $\mathbb{E}(\cdot)$. For a vector $x$, $x_+$ denotes its positive part, obtained by setting the non-positive components to zero. Additionally, $\lfloor\cdot\rfloor$ is the floor function, i.e., $\lfloor x\rfloor=\max\{n\in\mathbb{Z}|n\leq x\}$. Finally, $f^{(i)}(\cdot)$ denotes the $i$-fold composition of a function, and thus $f^{(i)}(\cdot) = f^{(i-1)}[f(\cdot)]$.

%%%%%%%%%%%%%%%%%%%%%%%%%%%%%%%%%%%%
\section{Problem Formulation}\label{chap:formulation}
\subsection{System Model}
Consider the following $N$ independent discrete-time linear time-invariant systems. Each system is observed by a sensor,
%(Fig.~\ref{fig:framework}),
 i.e., for $i = 1,\dots,N$:
\begin{equation*}
x_{k+1}^{[i]} 	= A_ix_k^{[i]} + w_k^{[i]},\qquad y_k^{[i]}	=C_ix_k + v_k^{[i]},
\end{equation*}
where $x_k^{[i]}\in\mathbb{R}^{n_i}$ is the state of the $i$-th system, $y_k^{[i]}\in\mathbb{R}^{m_i}$ is the noisy measurement obtained by the $i$-th sensor, $w^{[i]}_k\in\mathbb{R}^{n_i}$ and $v_k^{[i]}\in\mathbb{R}^{m_i}$ are mutually uncorrelated zero-mean Gaussian random variables with covariance $Q_i\succeq 0$ and $R_i\succ0$, respectively. The initial state $x^{[i]}_0$ is Gaussian distributed with mean zero and covariance $P_0^{[i]}\succeq0$, and is uncorrelated with $w^{[i]}_k$ and $v_k^{[i]}$ for all $k$. Assume $(A_i,\sqrt{Q_i})$ is stabilizable and $(A_i,C_i)$ is detectable.

\begin{figure}[!tbp]
	\centering
	\includegraphics[width=\linewidth]{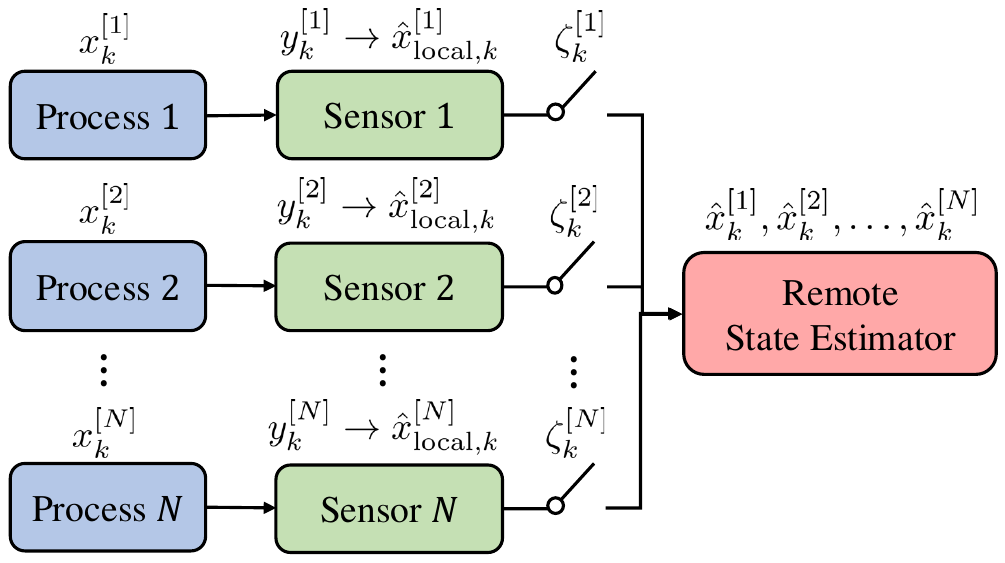}
	\caption{System diagram.}
	\label{fig:framework}
\end{figure}

We assume each sensor is capable of running a Kalman filter locally to compute the minimum mean square error (MMSE) estimate of $x^{[i]}_k$ based on its collected measurements, i.e.,
\begin{align*}
\hat{x}_{{\rm local},k}^{[i]} 			&\triangleq \mathbb{E}\left[x_k^{[i]}|y_0^{[i]},\dots,y_k^{[i]}\right],\\
\hat{P}_{{\rm local},k}^{[i]} 			&\triangleq \mathbb{E}\left[(x_k^{[i]}-\hat{x}_{{\rm local},k}^{[i]})(x_k^{[i]}-\hat{x}_{{\rm local},k}^{[i]})^T|y_0^{[i]},\dots,y_k^{[i]}\right],
\end{align*}
where $\hat{x}_{{\rm local},k}^{[i]}$ and $\hat{P}_{{\rm local},k}^{[i]} $ are the local estimate and the corresponding estimation error covariance, respectively.

For notational simplicity, define the following functions:
\begin{align*}
\tilde{g}_i(X)	&\triangleq X-XC_i^T[C_iXC_i^T+R_i]^{-1}C_iX,\\
h_i(X)		&\triangleq A_iXA_i^T+Q_i.
\end{align*}
According to~\cite{anderson2012optimal}, $\hat{P}_{{\rm local},k}^{[i]}$ converges exponentially to a steady-state $\bar{P}_{i}$, which is uniquely determined by solving the equation $\tilde{g}_i[ h_i(X)]=X$. Since we focus on asymptotic performance over an infinite time horizon, we assume, without loss of generality, that the Kalman filter reaches this steady state at $k=0$, i.e., $\hat{P}_{{\rm local},0}^{[i]}=\bar{P}_{i}$.

After computing the local estimate $\hat{x}_{{\rm local},k}^{[i]}$, the sensor transmits it to the remote state estimator via wireless communications. We assume that the communication is instantaneous and error-free. However, due to limited bandwidth, the sensors may not always be able to transmit their data. Let $\zeta_k^{[i]}\in\{0,1\}$ denote the binary variable that indicates whether the sensor sends its local estimate to the estimator: $\zeta_k^{[i]}=1$ if it transmits $\hat{x}_{{\rm local},k}^{[i]}$, and $\zeta_k^{[i]}=0$ otherwise. Moreover, denote $\tau_k^{[i]}$ as the time elapsed since the last transmission of the $i$-th sensor, i.e.,
\begin{equation*}
\tau_k^{[i]}\triangleq\min\{t\geq 0:\zeta^{[i]}_{k-t}=1\}.
\end{equation*}

Based on the above settings, the remote state estimator updates its MMSE estimate and the associated error covariance as follows~\cite{shi2010kalman}: $\hat{x}^{[i]}_k=\hat{x}^{[i]}_{{\rm local},k}$ and $P^{[i]}_k=\bar{P}_{i}$ if $\zeta^{[i]}_k=1$; $\hat{x}^{[i]}_k=A_i\hat{x}_{k-1}^{[i]}$ and $P^{[i]}_k=h_i(P^{[i]}_{k-1})=h_i^{(\tau_k^{[i]})}(\bar{P}_{i})$ otherwise.
\subsection{Bandwith-Constrained Communications}
As mentioned above, since the bandwidth is limited, the communications between the sensors and the remote state estimator are constrained. We consider two types of constraints.
 
\noindent{\bf Case 1} (Communication rate constraints): The total communication rate $r_i\triangleq\lim_{T\to\infty}\frac{1}{T+1}\sum_{k=0}^T\mathbb{E}[\zeta_k^{[i]}]$ is limited, i.e.,
\begin{equation}\label{eq:case1}
\sum_{i=1}^Nr_i\leq R,\qquad 0\leq r_i\leq 1,\quad i = 1,\dots,N,
\end{equation}
where $R>0$ can be interpreted as the total bandwidth resources shared by sensors.

\noindent{\bf Case 2} (Sensor activation constraints): At each time, only a limited number of sensors are allowed for transmission, i.e.,
\begin{equation}\label{eq:case2}
\sum_{i=1}^N\zeta_k^{[i]}\leq Z,\quad k=0,1,2,\ldots,
\end{equation}
where $Z\in\mathbb{N}_+$ is the maximum number of sensors allowed to transmit data simultaneously.

\begin{remark}{\rm
Intuitively speaking, both cases limit the number of non-zero elements of $\zeta_k^{[i]}$. In particular, Case 1 focuses on the long-term behavior of $\zeta_k^{[i]}$, while Case 2 addresses its behavior at each time $k$.
}\end{remark}

%%%%%%%%%%%%%%%%%%%%%%%%%%%%%
\subsection{Problem of Interest}
Denote the cost of the $i$-th sensor as:
\begin{equation*}
J_i(\bm{\zeta}_i )\triangleq\limsup_{T\to\infty}\frac{1}{T+1}\sum_{k=0}^T{\rm Tr}(P_k^{[i]}),
\end{equation*}
where $\bm{\zeta}_i 	\triangleq\{\zeta_0^{[i]},\dots,\zeta_k^{[i]},\dots\}$.

Most existing literature on sensor scheduling aims to maximize system efficiency by minimizing the total cost, i.e.,
\begin{equation}\label{eq:E}
\tag{E}
\min_{\bm \zeta} 	 \sum_{i=1}^N J_i({\bm \zeta}_i)\qquad {\rm s.t.}~ \text{\eqref{eq:case1} or \eqref{eq:case2}},
\end{equation}
where $\bm{\zeta}\triangleq\{\bm{\zeta}_1,\dots,\bm{\zeta}_N \}$. However, the cost function in~\eqref{eq:E} reflects the overall system performance and neglects fairness, potentially leading to an unbalanced allocation of resources among sensors. To prioritize fairness explicitly, a min-max criterion can be adopted, aiming to minimize the worst-case cost experienced by any individual sensor:
\begin{equation}\label{eq:F}
\tag{F}
\min_{\bm \zeta}\max_{i=1,\dots,N} 	 \sum_{i=1}^N J_i({\bm \zeta}_i)\qquad {\rm s.t.}~\text{\eqref{eq:case1} or \eqref{eq:case2}}.
\end{equation}

However, the above problems either focus on efficiency~\eqref{eq:E} or fairness~\eqref{eq:F}, which may not be desirable if we want to achieve trade-offs between both criteria. Fig.~\ref{fig:fair} provides an illustrative example of~\eqref{eq:E} and~\eqref{eq:F}.
\begin{figure}[!htbp]
	\centering
	\includegraphics[width=\linewidth]{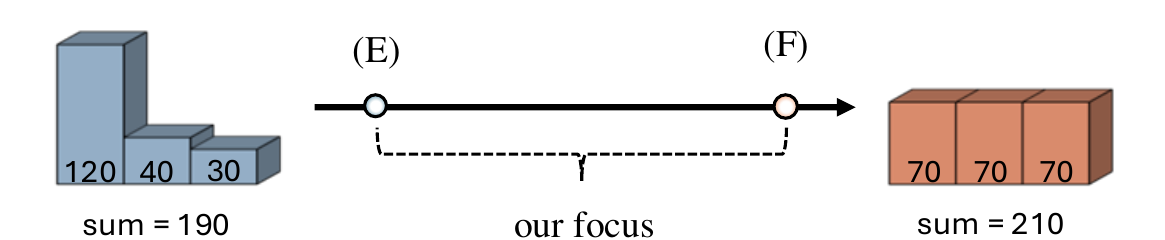}
	\caption{An illustrative example of~\eqref{eq:E} and~\eqref{eq:F}. In the case of~\eqref{eq:E}, the total cost $ \sum_{i=1}^N J_i(\bm{\zeta}_i)$ is minimized, but the distribution of costs across individuals is highly uneven, leading to unfair outcomes. Conversely,~\eqref{eq:F} achieves fairness by equalizing the costs among individuals, albeit at the expense of a higher overall cost.}
	\label{fig:fair}
\end{figure}

In the remainder of this article, we investigate how to design ${\bm\zeta}$ so that we can achieve a satisfying trade-off between efficiency and fairness under communication rate constraints~\eqref{eq:case1} or sensor activation constraints~\eqref{eq:case2}.
%%%%%%%%%%%%%%%%%%%%%%%%%%%%%
\section{Fairness Framework}\label{chap:fairness}
As mentioned, although the max-min fairness is believed to yield the ``fairest'' solution~\cite{wu2020max, 5461911}, it overlooks the overall system performance and may not be ideal when the estimation accuracy is also crucial. In order to address this issue, we instead employ the $q$-fairness framework~\cite{li2019fair}. Specifically, we resort to the following optimization problem:
\begin{equation}\label{eq:alpha}
\min_{\bm \zeta} \sum_{i=1}^N f_{q}(J_i({\bm \zeta}_i))\qquad {\rm s.t.}~\text{\eqref{eq:case1} or \eqref{eq:case2}},
\end{equation}
where 
\begin{equation*}
f_{q}(x)=\frac{x^{1+q}}{1+q},\quad 0\leq q<\infty.
\end{equation*}
\begin{remark}{\rm
The $q$-fairness framework is inspired by the well-known $\alpha$-fairness~\cite{mo2002fair}. For completeness, additional background on $\alpha$-fairness is provided in Appendix~\ref{apx:alpha}.
}\end{remark}
The function $f_{q}(x)$ achieves a trade-off between fairness and efficiency for the following reasons: (i) The function $f_{q}(x)$ is increasing in $x$, thus reflecting the efficiency of the system. (ii) The convex nature of $f_{q}(x)$ indicates {\it diminishing marginal returns} from the allocated resources. As shown by the yellow curve in {Fig.~\ref{fig:q}}, this property encourages ``richer'' sensors with lower costs to donate resources to ``poorer'' sensors with higher costs to achieve a better overall objective value, thereby promoting fairness.

%\begin{itemize}
%\item The function $f_{q}(x)$ is increasing in $x$, thus reflecting the efficiency of the system.
%\item The convex nature of $f_{q}(x)$ indicates {\it diminishing marginal returns} from the allocated resources. As shown by the yellow curve in Fig.~\ref{fig:q}, this property encourages ``richer'' sensors with lower costs to donate resources to ``poorer'' sensors with higher costs to achieve a better overall objective value, thereby promoting fairness.
%\end{itemize}

The value of $q$ affects the balance between fairness and efficiency in the following aspects: (i) When $q=0$, \eqref{eq:alpha} is equivalent to~\eqref{eq:E}. As $q$ increases, larger $J_i({\bm \zeta}_i)$ becomes more dominant in the objective value, causing the solution to \eqref{eq:alpha} to approach that of \eqref{eq:F}. (ii) Intuitively, increasing $q$ leads to a fairer solution, as it amplifies the diminishing marginal effect, encouraging ``richer'' sensors to donate more resources.
%\begin{itemize}
%\item When $q=0$, \eqref{eq:alpha} is equivalent to~\eqref{eq:E}. As $q$ increases, larger $J_i({\bm \zeta}_i)$ becomes more dominant in the objective value, causing the solution to \eqref{eq:alpha} to approach that of \eqref{eq:F}.
%\item Intuitively, increasing $q$ leads to a fairer solution, as it amplifies the diminishing marginal effect, encouraging ``richer'' sensors to donate more resources.
%\end{itemize}

{Fig.~\ref{fig:q}} provides a visual illustration of the function $f_q(x)$ for different values of $q$.
\begin{figure}[!htbp]
	\centering
	\includegraphics[width=\linewidth]{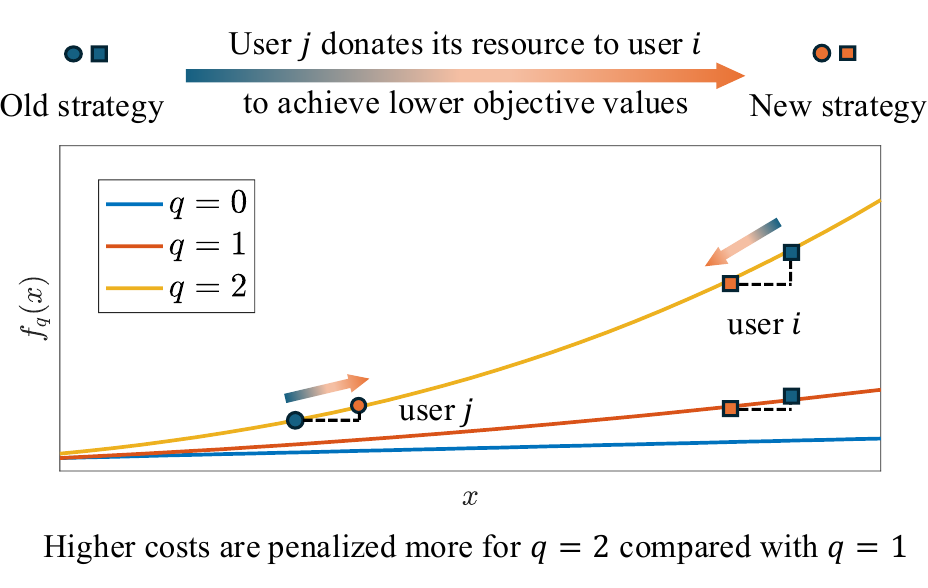}
	\caption{Illustration of the function $f_q(x)$.}
	\label{fig:q}
\end{figure}

\begin{remark}{\rm
While Fig.~\ref{fig:q} intuitively illustrates how varying $q$ impacts the overall trade-off, a formal theoretical analysis is beyond the scope of this work and remains an open question in the existing literature~\cite{mo2002fair,5461911}. Nevertheless, the numerical results in Section~\ref{chap:sim} align with our intuition that increasing $q$ leads to a fairer solution.
}\end{remark}

\begin{remark}{\rm
    The $q$-fairness framework can be readily extended to support sensor-specific penalties through personalized parameters $q_i$. More specifically, since a higher $q_i$ value imposes a greater penalty on higher costs, it can be assigned to a critical sensor to aggressively penalize its high estimation error. Importantly, the developed algorithms and theoretical results remain valid under this generalization.
}\end{remark}
%
%%%%%%%%%%%%%%%%%%%%
\section{Case 1: Communication Rate Constraints}\label{sec:bandwidth}
In this section, we provide fair scheduling strategies under communication rate constraints. Specifically, we solve~\eqref{eq:alpha} optimally under the constraint~\eqref{eq:case1}.

\subsection{Problem Decomposition}
To solve~\eqref{eq:alpha}, we decompose the problem into two steps: (i) determine the optimal scheduling policy under a given allocated communication rate $\tilde{r}_i$; (ii) 
investigate how to optimally allocate the communication rate $\tilde{r}_i$ under the constraints $\sum\tilde{r}_i\leq R$ and $0\leq\tilde{r}_i\leq 1$.

Given the communication rate $\tilde{r}_i$, the optimal policy for a single sensor is determined by the following problem:
\begin{equation}\label{eq:single}
\min_{{\bm \zeta}_i} 	J_i({\bm \zeta}_i)\qquad{\rm s.t.}~ r_i\leq\tilde{r}_i.
\end{equation}
\begin{lemma}[Theorem 4~\cite{chakravorty2017fundamental}]{\rm
The solution to~\eqref{eq:single} takes the following form:
\begin{equation}\label{eq:strategy}
\zeta_k^{[i]}=
\begin{cases}
0, 		& \text{if $\tau_k^{[i]}<\eta_i$};\\
0, 		& \text{with probability $1-p_i$, if $\tau_k^{[i]}=\eta_i$}; \\
1, 		& \text{with probability $p_i$, if $\tau_k^{[i]}=\eta_i$}; \\
1, 		& \text{if $\tau_k^{[i]}>\eta_i$}, 
\end{cases}
\end{equation}
where
\begin{equation}
\eta_i=\lfloor\frac{1}{\tilde{r}_i}-1\rfloor,\quad p_i=\eta_i+2-\frac{1}{\tilde{r}_i}. \tag*{$\square$}
\end{equation}
}\end{lemma}
To analyze the multi-sensor scenario, let $\beta_i=\lfloor\frac{1}{\tilde{r}_i}\rfloor\in\mathbb{N}_+$. It is easy to verify that
\begin{equation*}
\tilde{J}_i(\tilde{r}_i)=\left(1-\tilde{r}_i\beta_i\right){\rm Tr}[h_i^{(\beta_i)}(\bar{P}_i)]+\tilde{r}_i\sum_{j=0}^{\beta_i-1}{\rm Tr}[h_i^{(j)}(\bar{P}_i)].
\end{equation*}

While the cost  $\tilde{J}_i({r}_i)$ remains bounded for any $0< r_i\leq 1$ for systems with $\rho(A_i)\geq1$, this condition alone is insufficient to demonstrate the convergence of the primal-dual subgradient methods (to be introduced later). To ensure the boundness of the subgradients of $g(\bm{\tilde{r}})$ in~\eqref{eq:opt} and thus guarantee the convergence of the algorithm, we slightly modify the constraint $0\leq r_i\leq 1$ in~\eqref{eq:case1} to $\underline{r}_i\leq r_i\leq 1$, with $\underline{r}_i=0$ if $\rho(A_i)<1$ and $\underline{r}_i=\epsilon$ if $\rho(A_i)\geq1$, where $\epsilon$ is a positive but small number.

Let
\begin{equation*}
\begin{aligned}
\bm{\tilde{r}}&=\begin{bmatrix}\tilde{r}_1&\cdots&\tilde{r}_N\end{bmatrix},
\qquad\quad{\underline{\bm r}}=\begin{bmatrix}\underline{r}_1&\cdots&\underline{r}_N\end{bmatrix},\\
\mathcal{A} &= 
\begin{bmatrix}
\bm{1}_N &{I}_N &-{I}_N
\end{bmatrix}^T,\quad
{\bm b}=\begin{bmatrix}
R	&	\bm{1}_N^T		&\bm{\underline{r}}^T
\end{bmatrix}^T.
\end{aligned}
\end{equation*}
Then,~\eqref{eq:alpha} can be transformed into
\begin{equation}\label{eq:opt}
\min_{\tilde{\bm{r}}}g(\bm{\tilde{r}})\qquad
{\rm s.t.}						~\mathcal{A}\tilde{\bm{r}}\leq {\bm b},
\end{equation}
where $g(\bm{\tilde{r}})=\sum_{i=1}^Nf_{q}(\tilde{J}_i(\tilde{r}_i))$.
\subsection{Primal-Dual Subgradient Method}
To solve~\eqref{eq:opt}, we resort to the following augmented optimization problem:
\begin{equation*}
\min_{\tilde{\bm{r}}}g(\bm{\tilde{r}})+\frac{\alpha}{2}\|(\mathcal{A}\tilde{\bm{r}}-\bm{b})_+\|_2^2\qquad
{\rm s.t.}~\mathcal{A}\tilde{\bm{r}}\leq \bm{b},
\end{equation*}
where $\alpha>0$ is a regularization parameter. The regularization term penalizes constraint violations without altering the optimal solution. The associated Lagrangian is
\begin{equation*}
\mathcal{L}(\bm{\tilde{r}},{\bm \nu})	=g(\bm{\tilde{r}})+{\bm\nu}^T(\mathcal{A}\tilde{\bm{r}}- \bm{b})+\frac{\alpha}{2}\|(\mathcal{A}\tilde{\bm{r}}- \bm{b})_+\|_2^2,			
\end{equation*}
where ${\bm\nu}\geq0$ is the dual variable.

Define a mapping $T(\cdot,\cdot)$ by
\begin{equation*}
\begin{aligned}
T(\bm{\tilde{r}},{\bm \nu})&=
\begin{bmatrix}
\partial_{\tilde{\bm r}}\mathcal{L}(\bm{\tilde{r}},{\bm \nu})\\
-\partial_{\bm \nu}\mathcal{L}(\bm{\tilde{r}},{\bm \nu})
\end{bmatrix}\\
&=
\begin{bmatrix}
\partial g(\bm{\tilde{r}})+\mathcal{A}^T\bm{\nu}+\alpha\mathcal{A}^T(\mathcal{A}\tilde{\bm{r}}-\bm{b})_+\\
\bm{b}-\mathcal{A}\tilde{\bm{r}}
\end{bmatrix},
\end{aligned}
\end{equation*}
where $\partial g(\bm{\tilde{r}})$ is the subdifferential of $g(\tilde{\bm{r}})$. It is obvious that $\begin{bmatrix}\kappa_1&\cdots&\kappa_N\end{bmatrix}\in\partial g(\bm{\tilde{r}})$, where
\begin{equation*}
\kappa_i=\tilde{J}_i^q(\tilde{r}_i)\left[\sum_{j=0}^{\beta_i-1}{\rm Tr}[h^{(j)}(\bar{P}_i)]-\beta_i {\rm Tr}[h_i^{(\beta_i)}(\bar{P}_i)]\right].
\end{equation*}

Applying the primal-dual subgradient method, for each iteration $t$ of the algorithm, we update the primal and dual variables by
\begin{align}
\bm{\tilde{r}}(t+1)	&=\bm{\tilde{r}}(t)-w(t)\partial_{\tilde{\bm r}}\mathcal{L}(\bm{\tilde{r}}(t),{\bm \nu}(t)),\label{eq:primal_var}\\
			\bm{\nu}(t+1)		&=\bm{\nu}(t)-w(t)\partial_{\bm \nu}\mathcal{L}(\bm{\tilde{r}}(t),{\bm \nu}(t)),\label{eq:dual_var}
\end{align}
where $w(t)=\gamma(t)/\|T(\bm{\tilde{r}}(t),{\bm \nu(t)})\|_2$ is the step size such that
\begin{equation}\label{eq:step}
\gamma(t)>0,\qquad \sum_{t=1}^{\infty}\gamma(t)=\infty,\qquad \sum_{t=1}^{\infty}\gamma^2(t)<\infty.
\end{equation}

\begin{lemma}\label{lemma:converge}{\rm
For any step size $w(t)$ that satisfies~\eqref{eq:step}, the update rules~\eqref{eq:primal_var}-\eqref{eq:dual_var} converge, i.e.,
\begin{equation*}
\lim_{t\to\infty}g(\tilde{\bm r}(t))\to g(\tilde{\bm r}^{\star}), \qquad \lim_{t\to\infty}\|[\mathcal{A}\tilde{\bm r}(t)-{\bm b}]_+\|_2\to0,
\end{equation*}
where  $g(\tilde{\bm r}^{\star})$ is the optimal value of~\eqref{eq:opt}, and $\bm{\tilde{r}}^{\star}$ is any solution that achieves the optimal value.

}\end{lemma}
\begin{proof}
The proof relies on the convexity of~\eqref{eq:opt} and the boundedness of $g(\tilde{\bm{r}})$. Interested readers may refer to Appendix~\ref{apx:converge} for the complete proof.
\end{proof}
{Algorithm~\ref{alg:primal dual} summarizes the detailed procedure\footnote{As established in~\cite{nedic2009subgradient,goffin1977convergence}, the primal-dual subgradient algorithm achieves a sublinear convergence rate of $O(1/\sqrt{t})$. Its computational complexity per iteration is dominated by matrix multiplications yet avoids the more costly operation of matrix inversions, contributing to its efficiency.} for computing $\tilde{\bm r}^{\star}$}. Once $\tilde{\bm r}^{\star}$ is obtained, the scheduling strategy can be recovered by~\eqref{eq:strategy}. This yields a fair schedule $\bm{\zeta}$ that is the optimal solution to problem~\eqref{eq:alpha} under constraint~\eqref{eq:case1}.

\floatname{algorithm}{Algorithm}
\begin{algorithm}
	\caption{Primal-dual subgradient algorithm.}
	\label{alg:primal dual}
	\begin{algorithmic}[]
		\State Randomly initialize $\bm{\tilde{r}}(0)$ and $\bm{\nu}(0)$ and set $t \gets 0$
		\Repeat
			\State Compute $\bm{\tilde{r}}(t+1)$ and $\bm{\nu}(t+1)$ via~\eqref{eq:primal_var} and~\eqref{eq:dual_var}
			\State $t\gets t+1$
		\Until{convergence or maximum iteration numbers}
	\end{algorithmic}
\end{algorithm}

%%%%%%%%%%%%%%%%%%%%
\section{Case 2: Sensor Activation Constraints}\label{sec:activation}
\begin{figure*}[!t]
	\centering
	\includegraphics[width=\linewidth]{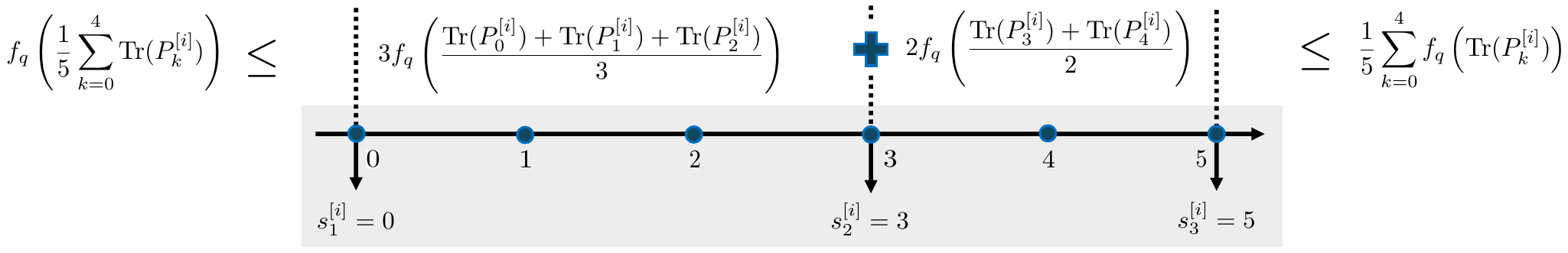}
	\caption{An illustrative example of the relaxation~\eqref{eq:relax} for arbitrary $i$-sensor with $T=4$, $s_1^{[i]}=0$, $s_2^{[i]}=3$ and $s_3^{[i]}=5$.}
	\label{fig:relax}
\end{figure*}
In this section, {we slightly relax~\eqref{eq:alpha} under the constraint~\eqref{eq:case2}} and propose two suboptimal algorithms based on MDP and greedy heuristics to solve it. Additionally, we evaluate their performance gaps compared to the optimal strategy.

The optimal solution yields a finite objective value since any communication rate $r_i>0$ leads to a finite $J_i({\bm \zeta}_i)$. Due to the continuity of $f_q(\cdot)$, we have
\begin{equation*}
\begin{aligned}
\sum_{i=1}^N f_q(J_i({\bm \zeta}_i))	&=\sum_{i=1}^N f_q\left(\limsup_{T\to\infty}\frac{1}{T+1}\sum_{k=0}^T{\rm Tr}(P_k^{[i]})\right)\\
							&=\limsup_{T\to\infty}\sum_{i=1}^N f_q\left(\frac{1}{T+1}\sum_{k=0}^T{\rm Tr}(P_k^{[i]})\right).
\end{aligned}
\end{equation*}

Define the time instants when the $i$-th sensor is scheduled to transmit:
\begin{equation*}
\begin{aligned}
s_1^{[i]}	&\triangleq\min\{k:k\in\mathbb{N},\zeta_k^{[i]}=1\},\\
s_j^{[i]}	&\triangleq\min\{k:k>s_{j-1}^{[i]},\zeta_k^{[i]}=1\},~j=2,3,\dots
\end{aligned}
\end{equation*}
Denote $\sigma_i(T)$ as the total number of transmissions by $i$-th sensor within time horizon $T$. Since $f_q(\cdot)$ is convex, we have
\begin{equation}\label{eq:relax}
f_q\left(\frac{1}{T+1}\sum_{k=0}^T{\rm Tr}(P_k^{[i]})\right)	\leq \frac{1}{T+1}\sum_{j=1}^{\sigma_i(T)}d_{j}^{[i]}f_q\left(\varrho^{[i]}_j\right),
\end{equation}
where
\begin{equation*}
\begin{aligned}
d_{j}^{[i]}&=
\begin{cases}
s_{j+1}^{[i]}-s_j^{[i]},	&\text{if $j\leq \sigma_i(T)-1$};\\
T-s_j^{[i]},			&\text{otherwise},
\end{cases}\\
\varrho^{[i]}_j&=\frac{1}{d_{j}^{[i]}}\sum_{k=0}^{d_{j}^{[i]}-1}{\rm Tr}\left(P_{k+s_j^{[i]}}^{[i]}\right).
\end{aligned}
\end{equation*}

Therefore, the original problem can be relaxed as
\begin{equation}\label{eq:relaxed}
\begin{aligned}
\min_{\bm \zeta} 	&\quad\limsup_{T\to\infty} \frac{1}{T+1}\sum_{i=1}^N\sum_{j=1}^{\sigma_i(T)}d_{j}^{[i]}f_q\left(\varrho^{[i]}_j\right)\\
{\rm s.t.}							&\quad \sum_{i=1}^N\zeta_k^{[i]}\leq Z,\quad \zeta_k^{[i]}\in\{0,1\}, \quad i=1,\dots,N.
\end{aligned}
\end{equation}
\begin{remark}{\rm
Using the convexity of $f_q(\cdot)$ yields the inequality $f_q\left(\frac{1}{T+1}\sum_{k=0}^T{\rm Tr}(P_k^{[i]})\right)\leq\frac{1}{T+1}\sum_{k=0}^Tf_q\left({\rm Tr}(P_k^{[i]})\right)$. Nevertheless, as shown in {Fig.~\ref{fig:q}}, fairness is promoted by the convex nature of $f_q(\cdot)$. A direct application of this inequality relaxes the problem into a linear form, thus compromising fairness. To address this issue, we segment the time sequence into subgroups at points where $\tau_k^{[i]}=0$ and apply the convexity inequality across these subgroups. An illustrative example of the relaxation is provided in Fig.~\ref{fig:relax}.
}\end{remark}
\subsection{MDP Formulation}
We solve the relaxed problem~\eqref{eq:relaxed} with an infinite-horizon average-cost MDP model. The formulation requires the following definitions.
\setlist[description]{font=\normalfont\itshape\textbullet\space}
\begin{description}
\item[State:] Denote $\phi_k$ as the state at time $k$:
\begin{equation*}
\phi_k\triangleq(\tau_k^{[1]},\dots,\tau_k^{[N]}),\quad\tau_k^{[i]}\in\mathbb{N}.
\end{equation*}
Therefore, the state space $\mathbb{S}$ is countably infinite.
\item[Action:] The action $a_k$ is the scheduling {decision}, i.e.,
\begin{equation*}
a_k\triangleq(\zeta_k^{[1]},\dots,\zeta_k^{[N]}).
\end{equation*}
The action space is
\begin{equation*}
\mathbb{A}=\left\{(\zeta^{[1]},\dots,\zeta^{[N]})|\sum_{i=1}^N\zeta^{[i]}\leq Z,~\zeta^{[i]}\in\{0,1\}\right\}.
\end{equation*}
\item[Transition Probability:] The transition probability from $\phi_{k}$ to $\phi_{k+1}$ given action $a_k$ is
\begin{equation*}
\mathbb{P}(\phi_{k+1}|\phi_k,a_k)=
\begin{cases}
1,	&\text{if $\zeta_k^{[i]}=1$, $\tau_{k+1}^{[i]}=0$};\\
		&\text{or $\zeta_k^{[i]}=0$, $\tau_{k+1}^{[i]}=\tau_{k}^{[i]}+1$};\\
0,	&\text{otherwise}.
\end{cases}
\end{equation*}
\item[One-Stage Cost Function:] The cost function is aligned with the objective of~\eqref{eq:relaxed}, {capturing the incremental cost at each time step}. It is independent of the action $a_k$ and is defined as $c(\phi_k)\triangleq\sum_{i=1}^Nc_i(\tau_k^{[i]})$, where
\begin{equation*}
\begin{aligned}
c_i(\tau_k^{[i]})\triangleq&(\tau_k^{[i]}+1)f_q\left(\frac{\sum_{j=0}^{\tau_k^{[i]}}{\rm Tr}[h^{(j)}(\bar{P}_i)]}{\tau_k^{[i]}+1}\right)\\
				&\quad-\tau_k^{[i]}f_q\left(\frac{\sum_{j=0}^{\tau_k^{[i]}-1}{\rm Tr}[h^{(j)}(\bar{P}_i)]}{\tau_k^{[i]}}\right).
\end{aligned}
\end{equation*}
\end{description}

A policy $\pi\in\Pi$ is a function that maps a state $\phi_k$ to an action $a_k$, where $\Pi$ is the set of admissible decisions. The average cost induced by policy $\pi$ is then defined as
\begin{equation*}
c_{\pi}(\phi_0)=\limsup_{T\to\infty}\frac{1}{T+1}\mathbb{E}_{\phi_0}^{\pi}\left[\sum_{k=0}^{T}c(\phi_k)\right],
\end{equation*}
where $\phi_0$ is the initial state, and the expectation is taken with respect to the policy $\pi$. We aim to find an optimal policy, i.e., $\pi^{\star}\in\Pi$ such that the average cost is minimized
\begin{equation*}
c_{\pi^{\star}}(\phi_0)\leq c_{\pi}(\phi_0),\quad \forall\phi_0\in\mathbb{S},\quad\forall\pi\in\Pi.
\end{equation*}
Specifically, we are interested in finding an optimal deterministic stationary policy. A policy $\pi$ is said to be deterministic if for each state $\phi_k$, it specifies a single action $a_k$ to take. It is said to be {stationary} if the policy is time-invariant and only depends on the current state $\phi_k$.

\begin{thm}\label{thm:bellman}{\rm
The optimal cost $c_{\pi^{\star}}$ satisfies the average cost Bellman equation, i.e.,
\begin{equation}\label{eq:bellman}
c_{\pi^{\star}}+V(\phi)=\min_{a\in\mathbb{A}}\{c(\phi)+\sum_{\phi'\in\mathbb{S}}\mathbb{P}(\phi'|\phi,a)V(\phi')\},
\end{equation}
where {$V(\cdot):\mathbb{S}\to\mathbb{R}$ is the relative value function}. Additionally, there exists a deterministic stationary policy $\pi^{\star}$, a constant $c_{\pi^{\star}}$ and a continuous function $V(\cdot)$ which satisfies the Bellman equation~\eqref{eq:bellman}. The optimal policy $a^{\star}(\phi)=\pi^{\star}(\phi)$ is determined by
\begin{equation*}
a^{\star}(\phi)=\arg\min_{a\in\mathbb{A}}\{c(\phi)+\sum_{\phi'\in\mathbb{S}}\mathbb{P}(\phi'|\phi,a)V(\phi')\}.
\end{equation*}
}\end{thm}
\begin{proof}{
The proof involves verifying the five conditions outlined in~\cite[Theorem 5.5.4]{hernandez2012discrete}. Interested readers may refer to Appendix~\ref{apx:bellman} for the complete proof.
}\end{proof}
The Bellman equation~\eqref{eq:bellman} can be solved by the relative value iteration algorithm~\cite{bertsekas2011dynamic}. With slight abuse of notations, let $V_t(\phi)$ represent the value of $\phi$ at the $t$-th iteration of the algorithm. Selecting an arbitrary but fixed state $\phi_f\in\mathbb{S}$, we obtain the following update rule for relative value iteration:
\begin{equation}\label{eq:value iteration}
\begin{aligned}
V_t(\phi)	&=\min_{a\in\mathbb{A}}\{c(\phi)+\sum_{\phi'\in\mathbb{S}}\mathbb{P}(\phi'|\phi,a)V_{t-1}(\phi')\}\\
		&\quad-\min_{a\in\mathbb{A}}\{c(\phi_f)+\sum_{\phi'\in\mathbb{S}}\mathbb{P}(\phi'|\phi_f,a)V_{t-1}(\phi')\}.
\end{aligned}
\end{equation}
\begin{lemma}\label{lemma:value}{\rm
The update rule~\eqref{eq:value iteration} converge to a solution to~\eqref{eq:bellman}, i.e., $\lim_{t\to\infty}V_t(\phi)\to V(\phi)$.
}\end{lemma}
\begin{proof}
The proof follows directly from~\cite[Proposition 4.3.2]{bertsekas2011dynamic}, thus omitted here.
\end{proof}

The following theorem demonstrates the existence of an optimal policy that converges to be asymptotically periodic~\cite{orihuela2014periodicity}. It provides a necessary condition for optimality and can be utilized to refine non-periodic suboptimal policies by interchanging sensor transmissions~\cite[Theorem 3, Example 1]{han2017optimal}. 
\begin{thm}[Asymptotic periodicity]\label{thm:period}{\rm There exists an optimal policy $\pi^{\star}$ converging to be periodic, i.e., under $\pi^{\star}$, $\exists L,M\in\mathbb{N}_+$ such that $a^{\star}_k=a^{\star}_{k+L}$ for all $k\geq M$.
}\end{thm}
\begin{proof}
{
The proof leverages the deterministic nature of $\pi^\star$ and the fact that $\tau_{k}^{[i]}$ resets to zero when $\zeta_k^{[i]}=1$. Interested readers may refer to Appendix~\ref{apx:period} for the complete proof.
}\end{proof}

Denote $\tilde{c}_{\#}$ as the optimal value of the original activation-constrained problem, i.e., problem~\eqref{eq:alpha} with constraint~\eqref{eq:case2}. The following lemma outlines the performance gap between the MDP solution to~\eqref{eq:relaxed} and the optimal solution to~\eqref{eq:alpha} .

\begin{lemma}\label{lemma:performance gap 1}{\rm
We have $0\leq c_{\pi^{\star}}-\tilde{c}_{\#}\leq c_{\pi^{\star}}-g(\tilde{\bm r}^{\star})$, where $g(\tilde{\bm r}^{\star})$ denotes the optimal value of~\eqref{eq:opt} with $R=Z$.
}\end{lemma}
\begin{proof}
Since $\tilde{c}_{\#}$ is the optimal value, obiviously $\tilde{c}_{\#}\leq c_{\pi^{\star}}$. Furthermore, when $\sum_{i=1}^N\zeta_k^{[i]}\leq Z$, the corresponding communication rate satisfies
\begin{equation*}
\begin{aligned}
&\sum_{i=1}^Nr_i	=\sum_{i=1}^N\lim_{T\to\infty}\frac{1}{T+1}\sum_{k=0}^T\mathbb{E}[\zeta_k^{[i]}]\\
			=&\lim_{T\to\infty}\frac{1}{T+1}\sum_{k=0}^T\sum_{i=1}^N\mathbb{E}[\zeta_k^{[i]}]\leq\lim_{T\to\infty}\frac{1}{T+1}\sum_{k=0}^TZ=Z.
\end{aligned}
\end{equation*}
Therefore, when $R=Z$, the sensor activation constraint \eqref{eq:case2} is a sufficient condition for the communication rate constraint \eqref{eq:case1}, and thus $g(\tilde{\bm r}^{\star})\leq\tilde{c}_{\#}$. Combining all these inequalities, we complete the proof.
\end{proof}
\begin{remark}{\rm

While $c_{\pi^\star}$ is not known in advance, the relative value iteration algorithm~\eqref{eq:value iteration} enables its determination by tracking the cost trajectory $\{c(\phi_1),\dots,c(\phi_k),\dots\}$. This makes $c_{\pi^{\star}}-g(\tilde{\bm r}^{\star})$ available during execution. Therefore, Lemma~\ref{lemma:performance gap 1} provides a real-time measure of the performance gap between the algorithm's output and the optimal value $\tilde{c}_{\#}$ during the implementation phase.
}\end{remark}

\subsection{Greedy Heuristics}
The state space grows with the number of sensors, rendering value iteration time-consuming with a large number of sensors. To address this, we introduce a suboptimal algorithm and analyze its performance gap relative to the optimal one.

The main idea of the greedy heuristic is to select the sensor that incurs the highest cost at each time $k$ if not selected. The complete algorithm is detailed in Algorithm~\ref{alg:greedy}.
\floatname{algorithm}{Algorithm}
\begin{algorithm}
	\caption{Greedy algorithm.}
	\label{alg:greedy}
	\begin{algorithmic}[]
		\State For each time $k$, initialize
		\begin{equation*}
		\mathcal{I}\gets\{1,\dots,N\}, \quad \zeta_k^{[i]}=0,~i\in\mathcal{I}
		\end{equation*}
				\For{$z$ from $1$ to $Z$} 
			\begin{align*}
			j\gets\arg\max_{j\in\mathcal{I}}c_j(\tau_k^{[j]}),\qquad \zeta_k^{[j]}\gets1,\qquad \mathcal{I}\gets\mathcal{I}\setminus j
			\end{align*}
				\EndFor
		\end{algorithmic}
\end{algorithm}
\begin{lemma}\label{lemma:period}{\rm
The schedule generated by Algorithm~\ref{alg:greedy} is periodic.
}\end{lemma}
\begin{proof}
The proof is similar to Therorem~\ref{thm:period}, thus omitted.
\end{proof}

Lemma~\ref{lemma:period} establishes that Algorithm~\ref{alg:greedy} yields a periodic schedule, a highly advantageous property that allows for easy implementation using a simple look-up table.

Denote $\tilde{c}_g$ as the cost obtained by Algorithm \ref{alg:greedy}. The following lemma describes the performance gap between Algorithm~\ref{alg:greedy} and the optimal schedule.

\begin{lemma}\label{lemma:performance gap 2}{\rm
We have $0\leq\tilde{c}_g-\tilde{c}_{\#}\leq \tilde{c}_g-g(\tilde{\bm r}^{\star})$, where $g(\tilde{\bm r}^{\star})$ denotes the optimal value of~\eqref{eq:opt} with $R=Z$.
}\end{lemma}
\begin{proof}
The proof is similar to Lemma~\ref{lemma:performance gap 1}, thus omitted.
\end{proof}

%%%%%%%%%%%%%%%%%%%%
\section{Simulation}\label{chap:sim}
In this section, we underscore the practical implications of our proposed $q$-fairness framework in achieving a flexible balance between system efficiency and fairness. Specifically, we demonstrate that increasing $q$ enhances fairness. Additionally, for sensor activation constraints, we validate the periodic structure (Theorem~\ref{thm:period} and Lemma~\ref{lemma:period}) and the performance gap (Lemma~\ref{lemma:performance gap 1} and Lemma~\ref{lemma:performance gap 2}) of the proposed algorithms.

To mathematically quantify the level of fairness, we introduce the following fairness measure.
\begin{definition}[Fairness measure~\cite{5461911}]\label{def}{\rm
 We say that ${\bm \zeta}$ is a fairer solution than ${\bm \zeta}'$ in terms of entropy if $\mathcal{H}[{\bm J}({\bm \zeta})]\geq \mathcal{H}[{\bm J}({\bm \zeta}')]$, where ${\bm J}({\bm \zeta})\triangleq\{J_1({\bm \zeta}_1),\dots,J_N({\bm \zeta}_N)\}$ and 
\begin{equation*}
\mathcal{H}[{\bm J}(\bm\zeta)]	\triangleq-\sum_{i=1}^N\frac{J_i({\bm \zeta}_i)}{\sum_i J_i({\bm \zeta}_i)}\log(\frac{J_i({\bm \zeta}_i)}{\sum_i J_i({\bm \zeta}_i)}).
\end{equation*}
}\end{definition}
\begin{remark}{\rm
Entropy, in information theory, measures the uncertainty of a random variable. Herein, we use it as a fairness measure. Intuitively, higher entropy indicates a more uniform distribution of the random variable, leading to a more balanced cost  $J_i({\bm \zeta}_i)$ across systems, and therefore a fairer allocation of resources.
}\end{remark}

{
\subsection{Communication Rate Constraints}\label{chap:sim_rate}
To enable a direct comparison with the max-min fair sensor schedules~\cite{wu2020max}, we employ the same system parameters, i.e.,
\begin{align*}
A_1&=\begin{bmatrix}1.2&0\\0 & 0 \end{bmatrix},\quad
A_2=\begin{bmatrix}1.1&1\\0 & 1 \end{bmatrix},\quad
A_3=\begin{bmatrix}1.2&1\\0 & 0.8 \end{bmatrix},\\
%%%
A_4&=\begin{bmatrix}0.8&0.6\\0 & 0.9 \end{bmatrix},\quad
A_5=\begin{bmatrix}0.3&1\\0 & 0.1 \end{bmatrix},\\
%%%
Q_1&=\begin{bmatrix}4&0\\0 & 1 \end{bmatrix},\quad
Q_2=\begin{bmatrix}1&0\\0 & 4 \end{bmatrix},\quad
Q_3=\begin{bmatrix}1&0\\0 & 4 \end{bmatrix},\\
%%%
Q_4&=\begin{bmatrix}16&0\\0 & 1 \end{bmatrix},\quad
Q_5=\begin{bmatrix}0.3&0\\0 & 1.2 \end{bmatrix},
\end{align*}
and $C_i=R_i=I_2$ for $i=1,\dots,5$. We set $R=Z=2$.

Fig.~\ref{fig:cost} demonstrates the costs of different sensors for different $q$ values. When $q=0$, the proposed $q$-fainress framework reduces to~\eqref{eq:E}, minimizing the total cost but resulting in an uneven cost distribution. As $q$ increases, the solution converges to the max-min fair schedule~\cite{wu2020max}, achieving a fair allocation of the communication rate. This trade-off between is further illustrated in Fig.~\ref{fig:cost_entr}: higher $q$ increases the total cost (reducing efficiency) but also increases the entropy (improving fairness).

\begin{figure}[!htbp]
	\centering
	\includegraphics[width=\linewidth]{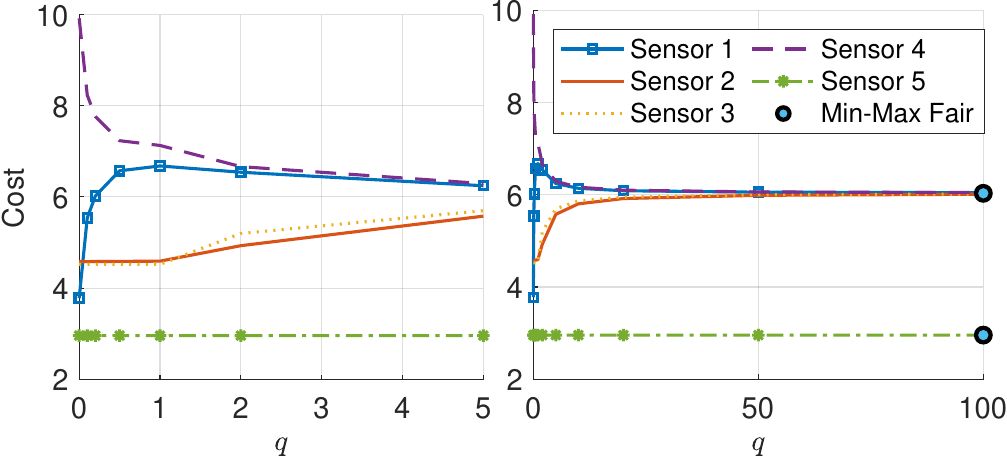}
	\caption{Cost of sensors for different $q$ values in Case 1. The max-min fair solution from~\cite{wu2020max} is shown for comparison.}
	\label{fig:cost}
\end{figure}

\begin{figure}[!htbp]
	\centering
	\includegraphics[width=\linewidth]{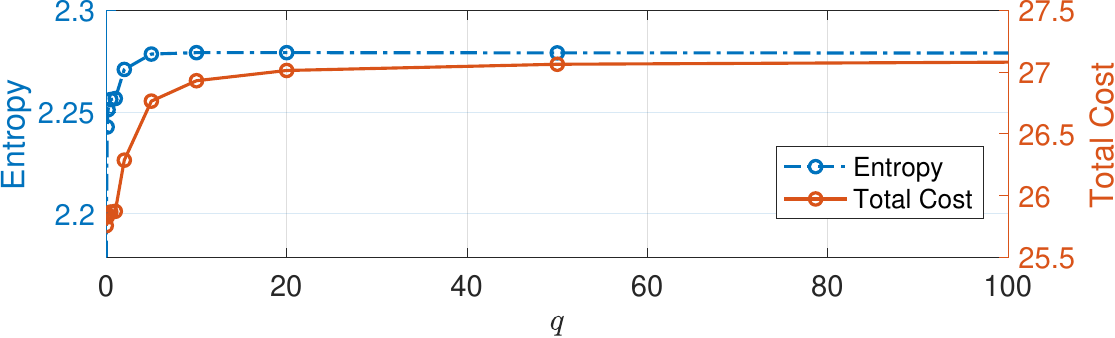}
	\caption{Total cost and entropy for different $q$ values in Case 1.}
	\label{fig:cost_entr}
\end{figure}
\subsection{Sensor Activation Constriants}
We modify some parameters while keeping all the others as defined in Section~\ref{chap:sim_rate}. The modified parameters are:
\begin{equation*}
A_3=\begin{bmatrix}1.1 & 0\\ 0 & 0\end{bmatrix},\quad
Q_3=\begin{bmatrix}1 &0\\ 0 &1\end{bmatrix},\quad
Q_5=\begin{bmatrix}2 &0\\ 0 &8\end{bmatrix},
\end{equation*}
and $R_3 = 10I_2$, $R_5=5I_2$.

\subsubsection{Periodic Structure} A realization of the action trajectory of the $2$nd sensor when $q=2$ is provided in Fig.~\ref{fig:period}. Note that both the MDP and greedy schedules exhibit asymptotic periodicity, which is aligned with Theorem~\ref{thm:period} and Lemma~\ref{lemma:period}.
\begin{figure}[!htbp]
	\centering
	\includegraphics[width=\linewidth]{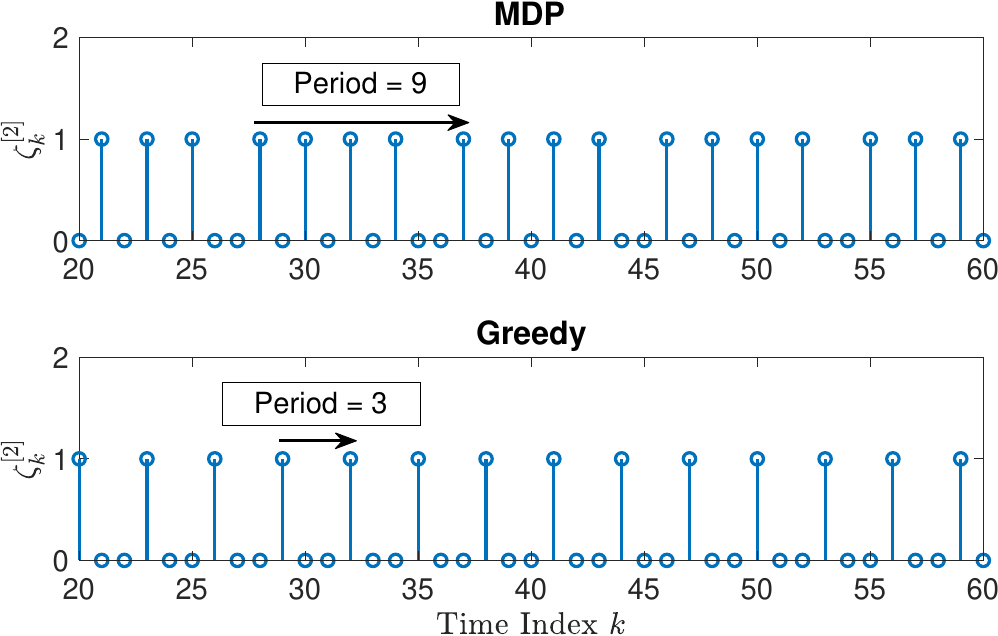}
	\caption{Schedule of the $2$nd sensor when $q=2$ in Case 2.}
	\label{fig:period}
\end{figure}

\subsubsection{MDP versus Greedy} Table~\ref{tb} provides a detailed comparison of the MDP method and the greedy algorithm. Similar to Case 1 (communication rate constraints), increasing the value of $q$ results in a rising trend in both entropy and total cost for both methods. Additionally, $c_{\pi^{\star}}/{g(\bm{\tilde{r}}^{\star})}\geq 1$ and ${\tilde{c}_g}/{g(\bm{\tilde{r}}^{\star})}\geq 1$, which are aligned with the results in Lemma~\ref{lemma:performance gap 1} and Lemma~\ref{lemma:performance gap 2}.}
\begin{table}[!htbp]
\centering
\caption{Comparison between MDP and the greedy algorithm in Case 2. The right column displays the relative performance, defined as $c_{\pi^{\star}}/{g(\bm{\tilde{r}}^{\star})}$ for MDP and ${\tilde{c}_g}/{g(\bm{\tilde{r}}^{\star})}$ for the greedy heuristics.}
\begin{tabular}{c|cc|cc|cc}
\hline
\multirow{2}{*}{$q$} & \multicolumn{2}{c|}{Entropy} & \multicolumn{2}{c|}{Total Cost} & \multicolumn{2}{c}{Relative} \\ \cline{2-7} 
                   & MDP         & Greedy         & MDP          & Greedy           & MDP & Greedy \\ \hline
0                  & 2.16           & 2.27        & 39.91            & 41.70         &1.01              & 1.05            \\
0.5              & 2.24          & 2.27         & 39.95           & 41.70         & 1.02           & 1.08            \\
2                  &  2.29          & 2.29         &  40.50            & 42.01       &  1.03              & 1.14           \\ 
20                  &  2.31           & 2.29         &  41.75            & 42.01        &  4.23              & 6.34            \\ \hline
\end{tabular}
\label{tb}
\end{table}

%%%%%%%%%%%%%%%%%%%%
\section{Conclusion}\label{chap:conclusion}
In this article, we addressed the fair sensor scheduling problem with bandwidth constraints using the $q$-fairness framework to balance efficiency and fairness. We examined two communication constraints: communication rate constraints and sensor activation constraints. We presented the optimal strategy for communication rate constraints and proposed two suboptimal algorithms for sensor activation constraints, along with an analysis of their performance gaps compared to the optimal policy. {Future work will investigate the theoretical impact of $q$ and fair sensor scheduling under event-triggered mechanisms.}

\section{Acknowledgement}
We would like to thank Dr.~Shenghui Wu for valuable discussions on this work.

\appendices
\section{Preliminaries on $\alpha$-fairness}\label{apx:alpha}
The $\alpha$-fairness framework~\cite{mo2002fair} is a popular approach in resource allocation problems. It defines an overall utility function\footnote{The utility functions can be regarded as the negative of the cost functions.} as the sum of individual utilities, calculated by
\begin{equation*}
f_{\alpha}(x)=
\begin{cases}
\log(x), &\text{if $\alpha=1$};\\
\frac{x^{1-\alpha}}{1-\alpha},&\text{if $\alpha\in[0,1)\cup(1,\infty)$},
\end{cases}
\end{equation*}
where $0\leq \alpha < \infty$, and $x$ is the allocated resources of each user. The goal is to find an allocation strategy that maximizes the overall utility. By adjusting the value of $\alpha$, the framework balances fairness and efficiency~\cite{5461911}. In particular, it results in no fairness when $\alpha=0$ and harmonic fairness~\cite{dashti2013harmonic} when $\alpha=2$. In the multiclass fluid model~\cite{mo2002fair}, it achieves proportional fairness~\cite{kelly1997charging} for $\alpha=1$ and max-min fairness~\cite{radunovic2007unified} as $\alpha\to\infty$. %Although these may not hold in general models, the $\alpha$-fairness framework remains useful for balancing fairness and efficiency. 

Unlike $\alpha$-fairness that captures efficiency by {\it maximizing utilities}, our $q$-fariness approach focuses on {\it minimizing costs}. 
\section{Proof of Lemma~\ref{lemma:converge}}\label{apx:converge}
Since $f_q(\cdot)$ is convex and non-decreasing and is composed with $\tilde{J}_i(\tilde{r}_i)$, which is piecewise linear~\cite{chakravorty2017fundamental}, the resulting optimization problem~\eqref{eq:opt} is convex.

Denote ${\bm z}=[\tilde{\bm r}^T~\bm{\nu}^T]^T$ and let $\bm{z}^{\star}=[(\tilde{\bm r}^{\star})^T~(\bm{\nu}^{\star})^T]^T$ be any pair of optimal variables. By substituting~\eqref{eq:primal_var}-\eqref{eq:dual_var}, we obtain:
\begin{equation*}
\begin{aligned}
	&\|\bm{z}(t+1)-\bm{z}^{\star}\|_2^2\\
=	&\|\bm{z}(t)-\bm{z}^{\star}\|_2^2+w^2(t)\|T(\bm{z}(t))\|_2^2\\
	&\quad-2w(t)T(\bm{z}(t))^T[\bm{z}(t)-\bm{z}^{\star}]\\
=	&\|\bm{z}(t)-\bm{z}^{\star}\|_2^2+\gamma^2(t)-2\gamma(t)\frac{T^T(\bm{z}(t))}{\|T(\bm{z}(t))\|_2}(\bm{z}(t)-\bm{z}^{\star}).
\end{aligned}
\end{equation*}
Summing over $t$, we arrive at
\begin{equation}\label{eq:sum}
\begin{aligned}
&\|\bm{z}(t+1)-\bm{z}^{\star}\|_2^2+\sum_{j=1}^t\gamma_j\frac{T^T(\bm{z}(j))}{\|T(\bm{z}(j))\|_2}(\bm{z}(j)-\bm{z}^{\star})\\
&=\|\bm{z}(1)-\bm{z}^{\star}\|_2^2+\sum_{j=1}^t\gamma^2(j).
\end{aligned}
\end{equation}
%%%%%
%\begin{comment}
%Since $\mathcal{A}\bm{\tilde{r}}^{\star}=\bm{b}$, we have
%\begin{equation}\label{eq:limit}
%\begin{aligned}
%&T^T(\bm{z}(t))(\bm{z}(t)-\bm{z}^{\star})\\
%\geq&\mathcal{L}(\bm{\tilde{r}}(t),\bm{\nu}^{\star})-\mathcal{L}(\bm{\tilde{r}}^{\star},\bm{\nu}^{\star})+\frac{\alpha}{2}\|(\mathcal{A}\bm{\tilde{r}}(t)-\bm{b})_+\|_2^2\geq0,
%\end{aligned}
%\end{equation}
%where the first inequality follows from the properties of subgradients. Thus, the second term on the left-hand side of~\eqref{eq:sum}, i.e., $s(t)=\sum_{j=1}^t\gamma_j{T^T(\bm{z}(j))}(\bm{z}(j)-\bm{z}^{\star})/\|T(\bm{z}(j))\|_2$ is non-negative. Since $\|\bm{z}(1)-\bm{z}^{\star}\|_2^2<\infty$ and $\sum_{j=1}^t\gamma^2(j)<\infty$, $s(t)$ is bounded too. On the other hand, since $\sum_{t=1}^{\infty}\gamma(t)=\infty$ and $\|T(\bm{z}(t))\|_2<\infty$, we have $\lim_{t\to\infty}T^T(\bm{z}(t))(\bm{z}(t)-\bm{z}^{\star})=0$. Taking limit of both sides of~\eqref{eq:limit}, we obtain $\lim_{t\to\infty}\|(\mathcal{A}\bm{\tilde{r}}(t)-\bm{b})_+\|_2=0$ and $\lim_{t\to\infty}\mathcal{L}(\bm{\tilde{r}}(t),\bm{\nu}^{\star})=\mathcal{L}(\bm{\tilde{r}}^{\star},\bm{\nu}^{\star})=g(\tilde{\bm{r}}^{\star})$,
%which completes the proof.
%\end{comment}
%%%
Since $\mathcal{A}\bm{\tilde{r}}^{\star}=\bm{b}$, we have
\begin{equation}\label{eq:limit}
\begin{aligned}
&T^T(\bm{z}(t))(\bm{z}(t)-\bm{z}^{\star})\\
\geq& g(\bm{\tilde{r}}(t))-g(\bm{\tilde{r}}^{\star})+(\bm{\nu}^{\star})^T(\mathcal{A}\bm{\tilde{r}}(t)-\bm{b})+\alpha\|(\mathcal{A}\bm{\tilde{r}}(t)-\bm{b})_+\|_2^2\\
=&\mathcal{L}(\bm{\tilde{r}}(t),\bm{\nu}^{\star})-\mathcal{L}(\bm{\tilde{r}}^{\star},\bm{\nu}^{\star})+\frac{\alpha}{2}\|(\mathcal{A}\bm{\tilde{r}}(t)-\bm{b})_+\|_2^2\geq0,
\end{aligned}
\end{equation}
where the first inequality follows from the properties of subgradients.

Thus, the second term on the left-hand side of~\eqref{eq:sum}, i.e.,
\begin{equation*}
s(t)=\sum_{j=1}^t\gamma_j\frac{T^T(\bm{z}(j))}{||T(\bm{z}(j))||_2}(\bm{z}(j)-\bm{z}^{\star}),
\end{equation*}
is non-negative. Since $\|\bm{z}(1)-\bm{z}^{\star}\|_2^2<\infty$ and $\sum_{j=1}^t\gamma^2(j)<\infty$, $s(t)$ is bounded too. On the other hand, since
\begin{equation*}
\sum_{t=1}^{\infty}\gamma(t)=\infty,\quad \|T(\bm{z}(t))\|_2<\infty,
\end{equation*}
we have $\lim_{t\to\infty}T^T(\bm{z}(t))(\bm{z}(t)-\bm{z}^{\star})=0$. Taking limit of both sides of~\eqref{eq:limit}, we obtain 
\begin{equation*}
\begin{aligned}
&\lim_{t\to\infty}\mathcal{L}(\bm{\tilde{r}}(t),\bm{\nu}^{\star})=\mathcal{L}(\bm{\tilde{r}}^{\star},\bm{\nu}^{\star})=g(\tilde{\bm{r}}^{\star}),\\
&\lim_{t\to\infty}\|(\mathcal{A}\bm{\tilde{r}}(t)-\bm{b})_+\|_2=0,
\end{aligned}
\end{equation*}
which completes the proof.

%%--------------------------------------
\section{Proof of Theorem~\ref{thm:bellman}}\label{apx:bellman}
Before presenting the proof of Theorem~\ref{thm:bellman}, we first introduce the following lemmas.
\begin{lemma}[Lemma 5.3.1~\cite{hernandez2012discrete}]\label{lemma:pre}{\rm
Let $\{c_k\}$, $k\in\mathbb{N}$ be a sequence of nonnegative numbers. Then, the following inequalities hold:
\begin{equation*}
\begin{aligned}
\liminf_{T\to\infty}&\frac{1}{T+1}\sum_{k=0}^{T}c_k\leq\liminf_{\delta\to1^-}(1-\delta)\sum_{k=0}^{\infty}\delta^kc_k\\
&\leq\limsup_{\delta\to1^-}(1-\delta)\sum_{k=0}^{\infty}\delta^kc_k\leq\limsup_{T\to\infty}\frac{1}{T+1}\sum_{k=0}^{T}c_k.
\end{aligned}
\end{equation*}
}\end{lemma}
\begin{lemma}\label{lemma:bounded cost}{\rm
No policy $\pi$ that results in an unbounded average cost $c_{\pi}(\phi_0)$ can be optimal.
}\end{lemma}
\begin{proof}
It is straightforward to verify that there exists a policy under which each sensor transmits its local estimate to the remote state estimator at least once every $N$ steps. Such a policy ensures that ${\rm Tr}(P_{k}^{[i]})$ remains bounded, leading to a bounded average cost. Therefore, any policy $\pi$ that induces an unbounded average cost cannot be optimal.
\end{proof}

We are now prepared to present the proof of Theorem~\ref{thm:bellman}.

Given a policy $\pi$, for $\phi_0=\phi$, the discounted cost is
\begin{equation*}
V_{\delta}(\pi,\phi)=\sum_{k=0}^{\infty}\delta^k\mathbb{E}^{\pi}_{\phi}[c(\phi_k)],
\end{equation*}
where $\delta\in(0,1)$ is the discount factor. Denote $V_{\delta}^\star(\phi)=\inf_{\pi\in{\Pi}}V_{\delta}(\pi,\phi)$.

According to to~\cite[Theorem 5.5.4 and Remark 5.5.2]{hernandez2012discrete}, to prove Theorem~\ref{thm:bellman}, only the following items need to be verified:
\begin{enumerate}[label=(\Roman*)]
\item The one-stage cost function $c(\phi)$ is lower semicontinuous and nonnegative.
\item The probability $\mathbb{P}(\phi'|\phi,a)$ is strongly continous.
\item There exists a state $\psi\in\mathbb{S}$ along with constants $\underline{\delta}\in(0,1)$ and $U\geq0$ such that

	(a) $(1-\delta)V_{\delta}^\star(\psi)\leq U,~\forall\delta\in[\underline{\delta},1)$.
	
	Moreover, there exists a constant $L\geq0$ and a nonnegative measurable function $l(\cdot)$ on $\mathbb{S}$ such that
	
	(b) $-L\leq V_{\delta}^\star(\phi)-V_{\delta}^\star(\psi)\leq l(\phi),~\forall \phi\in\mathbb{S},~\delta\in[\underline{\delta},1)$.
\item For any $\phi\in\mathbb{S}$ and $a\in\mathbb{A}$, the above function $l(\cdot)$ satisfies $\sum_{\phi'\in\mathbb{S}}l(\phi')\mathbb{P}(\phi'|\phi,a)<\infty$.
\item The state space $\mathbb{S}$ is a denumerable set.
\end{enumerate}

It is easy to verify that $c(\phi)$ is lower bounded by zero. Moreover, the state space $\mathbb{S}$ is countably infinite and the action space $\mathbb{A}$ is finite. Therefore, (I)(II) and (V) hold trivially.

Lemma~\ref{lemma:pre} establishes the following inequality by letting $c_k=\mathbb{E}_{\phi}^{\pi^\star}[c(\phi_k)]$:
\begin{equation*}
\limsup_{\delta\to1^-}(1-\delta)V_{\delta}^\star(\phi)\leq\limsup_{T\to\infty}\frac{1}{T+1}\sum_{k=0}^{T}\mathbb{E}_{\phi}^{\pi^\star}[c(\phi_k)].
\end{equation*}
By Lemma~\ref{lemma:bounded cost}, $\limsup_{T\to\infty}\frac{1}{T+1}\sum_{k=0}^{T}\mathbb{E}_{\phi}^{\pi^\star}[c(\phi_k)]$ is bounded. Therefore, for any state $\phi$, there exists $U\geq 0$ such that
\begin{equation}\label{eq:U}
(1-\delta)V_{\delta}^\star(\phi)\leq U,
\end{equation}
implying (III-a) holds for any state $\phi$ and, trivially, for $\psi$.

We pick $\psi=(\tau^{[1]},\dots,\tau^{[N]})$,
where
\begin{equation*}
\tau^{[i]}=\begin{cases}
\lfloor \frac{i-1}{Z}\rfloor,	&\text{if $i\leq\lfloor\frac{N}{Z}\rfloor Z$};\\
\lfloor \frac{N-1}{Z}\rfloor,	&\text{otherwise}.
\end{cases}
\end{equation*}
Let $\underline{k}=\lfloor\frac{N}{Z}\rfloor$. It can be verified that starting from any state $\phi$, there exists a policy $\hat{\pi}$ such that $\psi$ is reached in at most $\underline{k}+1$ steps\footnote{The policy $\hat{\pi}$ can be constructed as: $\zeta_{\underline{k}}^{[i]}=1$ for $1\leq i\leq Z$; $\zeta_{\underline{k}-1}^{[i]}=1$ for $Z+1\leq i\leq 2Z$;\dots;$\zeta_0^{[i]}=1$ for $i\geq\underline{k}Z+1$.}. Then, for any state $\phi$, the following holds:
\begin{equation*}
\begin{aligned}
V_{\delta}^{*}(\phi)	&\leq\mathbb{E}_{\phi}^{\hat{\pi}}\left[\sum_{k=0}^{\underline{k}-1}\delta^kc(\phi_k)\right]+\mathbb{E}^{\pi^\star}_{\phi}\left[\sum_{k=\underline{k}}^{\infty}\delta^{k}c(\phi_k)|\phi_{\underline{k}}=\psi\right]\\
				&=\mathbb{E}_{\phi}^{\hat{\pi}}\left[\sum_{k=0}^{\underline{k}-1}\delta^kc(\phi_k)\right]+\delta^{\underline{k}}V_{\delta}^{*}(\psi).
\end{aligned}
\end{equation*}
Obviously, the first term is bounded. Therefore, letting $l(\phi)=\mathbb{E}_{\phi}^{\hat{\pi}}\left[\sum_{k=0}^{\underline{k}-1}\delta^kc(\phi_k)\right]$, we obtain 
\begin{equation}\label{eq:3b-1}
\begin{aligned}
V_{\delta}^{*}(\phi)	&\leq l(\phi)+\delta^{\underline{k}}V_{\delta}^{*}(\psi)\leq l(\phi)+V_{\delta}^{*}(\psi).
\end{aligned}
\end{equation}

On the other hand, denote $\underline{v}\triangleq\inf_{\phi\in\mathbb{S}}V^\star_{\delta}(\phi)$. Starting from an arbitrary $\phi_0$, consider a policy (not necessarily optimal) under which it achieves $V_\delta^\star(\phi_{\bar{k}})\leq\underline{v}+\Delta$ at some time $\bar{k}$, where $\Delta\geq0$ is fixed value. From \cite[Lemma 4.1]{schal1993average}, we have
\begin{equation*}
0\leq V^\star_\delta(\phi)-\underline{v}\leq\mathbb{E}_{\phi}^{\pi^\star}\left[\sum_{k=0}^{\bar{k}-1}\delta^kc(\phi_k)\right].
\end{equation*}
Therefore,
\begin{equation}\label{eq:3b-2}
V^\star_\delta(\phi)-V^\star_\delta(\psi)\geq-\mathbb{E}_{\phi}^{\pi^\star}\left[\sum_{k=0}^{\bar{k}-1}\delta^kc(\phi_k)\right].
\end{equation}
Since $\bar{k}$ is bounded~\cite{nourian2014optimal}, $\mathbb{E}_{\phi}^{\pi^\star}\left[\sum_{k=0}^{\bar{k}-1}\delta^kc(\phi_k)\right]$ is bounded. Combining~\eqref{eq:3b-1} and~\eqref{eq:3b-2}, we complete the proof of (III-b).

The condition (IV) is satisfied trivially since, for any state $\phi$, the set of possible next states $\phi'$ is finite.

With all conditions (I)–(V) now verified, the proof of Theorem~\ref{thm:bellman} is complete.
%-----------------------------

\section{Proof of Theorem~\ref{thm:period}}\label{apx:period}
Under the optimal policy $\pi^\star$, each sensor $i$ falls into one of two categories:
\begin{itemize}
\item Quiet sensor: After some finite time, $\zeta_k^{(i)}=0$ for all $k$.
\item Transmitting sensor: $\zeta_{k}^{[i]}=1$ occurs infinitely often.
\end{itemize}
Let $\mathcal{I}_{q}$ and $\mathcal{I}_{t}$ denote the set of quiet and transmitting sensors, respectively, where $\mathcal{I}_{q}\cup\mathcal{I}_{t}=\{1,\dots,N\}$. Since sensors in  $\mathcal{I}_{q}$ are never chosen to transmit, we can focus solely on $\mathcal{I}_{t}$ without altering the problem. Specifically, the task reduces to selecting at most $Z$ sensors to transmit out from $\mathcal{I}_t$.

For the reduced problem, $\tau_k^{[i]}=0$ every time $\zeta_k^{[i]}=1$, which implies the induced Markov chain has a finite state space. Moreover, since $\pi^{\star}$ is deterministic, there must be a recurrent state that the system re-enters after a finite number of steps. This ensures that the system follows a periodic trajectory.

On the other hand, for $\mathcal{I}_{q}$, since $\zeta_k^{(i)}=0$ for all $k$, it is periodic straightforwardly.

Since both for sensors in  $\mathcal{I}_{t}$ and $\mathcal{I}_{q}$, the action displays a periodic nature, we complete the proof of Theorem~\ref{thm:period}.

%%%%%%%%%%%%%%%%%%%
\bibliographystyle{IEEEtran}
\bibliography{ref.bib}

\end{document}